\documentclass[11pt]{article}

\usepackage[margin=1in]{geometry}

\usepackage{amsmath,amssymb,amsfonts}
\usepackage{amsthm}

\usepackage{longtable}
\usepackage{booktabs}

\usepackage{url}
\usepackage[colorlinks=true,citecolor=blue,linkcolor=blue,urlcolor=blue]{hyperref}

\usepackage{authblk}   

\theoremstyle{plain}
\newtheorem{theorem}{Theorem}[section]
\newtheorem{lemma}[theorem]{Lemma}
\newtheorem{fact}[theorem]{Fact}
\newtheorem{corollary}[theorem]{Corollary}

\theoremstyle{definition}

\newtheorem{example}[theorem]{Example}

\theoremstyle{remark}
\newtheorem{remark}[theorem]{Remark}

\title{Resolution of an Open Problem on Quasi-Cyclic Codes over $\mathbb{Z}_4$ and New Quaternary Linear Codes}

\author[1]{Nuh Aydin}
\author[2]{Aditya Tyagi}
\affil[1]{Kenyon College, Gambier, OH 43022, USA}
\affil[2]{IIT Gandhinagar, Gandhinagar, GJ 382055, India}

\date{}

\begin{document}
\maketitle

\begin{abstract}
Given a cyclic code $C_g = \langle g(x) \rangle$ of odd length $m$ over
$\mathbb{Z}_4$, one common way to build a quasi-cyclic (QC) code is to pick
$f_1, \dots, f_\ell \in \mathbb{Z}_4[x]$ and let
$C = \langle (f_1 g, \dots, f_\ell g) \rangle$.
Because $\mathbb{Z}_4$ is not a field, the type of $C_g$ ($4^{k_1}2^{k_2}$)  is not necessarily inherited by $C$. Determining conditions under which the type of $C_g$ is inherited by $C$ was posed as an open problem recently in \cite{AydinLuOnta2023}. In this paper, we settle this
problem. We first give two sufficient conditions for the type to be
preserved: one requires a single $f_i$ to be coprime to $x^m-1$ over
$\mathbb{F}_2$, the other only requires the $f_i$ to be jointly
coprime to it. Neither condition is necessary in general. Using the fact that
$x^m-1$ is squarefree for odd $m$, we decompose
$\mathbb{Z}_4[x]/\langle x^m-1\rangle$ into a product of finite chain rings using Chinese remainder theorem (CRT)
and derive a condition on $f_1, \dots, f_\ell$ that is both necessary
and sufficient for $C$ to match the type of $C_g$. This condition depends only on
the irreducible factors of $x^m-1$ where $g$ does not already vanish.
This also yields a simple test for when $C$ is a free
$\mathbb{Z}_4$-module. Finally, we report many new QC codes over $\mathbb{Z}_4$, found by  computer searches using Magma software~\cite{Magma1997} guided by this criterion, with Lee distances greater than previously known codes of the same type.
\end{abstract}


\section{Introduction and Motivation}

As a generalisation of cyclic codes, quasi-cyclic (QC) codes are a
prominent class of linear codes (see
\cite{HuffmanPless2003,Wan1997,Dougherty2017}). They have been studied extensively since the foundational
work of \cite{ChenPeterson1969,Weldon1970,Kasami1974}, and a large number of new (record breaking) codes have been obtained from the class of QC codes over the last several decades. For a small sample of many publications that reported new linear codes obtained from QC codes see \cite{LingSole2003,chen2022,AkreAydinHarrington2023}. In recent years, QC codes on finite rings have also received much attention (see \cite{HammonsEtAl1994,BhaintwalWasan2009,Cao2018,Gao2014}). One of the earlier works on QC codes over finite rings was \cite{AydinRayChaudhuri2002} in which QC codes over $\mathbb{Z}_4$ were used to
construct good binary codes via the Gray map, including new nonlinear binary codes
(see~\cite{AydinRayChaudhuri2002}).

A commonly studied type of QC codes over fields or rings is obtained as follows. In the case of codes over $\mathbb{Z}_4$, we start with a cyclic code
\[
C_g=\langle g(x) \rangle
\]
of odd length $m$ over $\mathbb{Z}_4$. Then choose polynomials
\[
f_1(x),\ldots,f_\ell(x)\in\mathbb{Z}_4[x],
\]
and consider the $1$-generator QC code
\[
C=
\left\langle
\big(f_1(x)g(x),\ldots,f_\ell(x)g(x)\big)
\right\rangle.
\]
Here the resulting QC code has length $m\ell$ and index $\ell$.

Since $\mathbb{Z}_4$ is not a field, a cyclic code (or more generally a linear code) over $\mathbb{Z}_4$
need not be free as a $\mathbb{Z}_4$-module. Nevertheless, by the
structure theorem for finitely generated modules over $\mathbb{Z}_4$ (see~\cite{McDonald1974}),
its type is well defined. Thus, if $C_g$ has type
\[
4^{k_1}2^{k_2},
\]
then
\[
C_g\cong \mathbb{Z}_4^{k_1}\oplus\mathbb{Z}_2^{k_2},
\qquad
|C_g|=4^{k_1}2^{k_2}.
\]

When the above construction is used to obtain a QC code from a cyclic
code, the type need not be preserved. This leads to the following
question, which was stated as an open problem in
\cite{AydinLuOnta2023}:
\[
\textit{Under what conditions does $C$ have the same type as $C_g$?}
\]

In this work, we answer this question. We first give two sufficient
conditions. Theorem 3.1 assumes that at least one of the polynomials
$f_1,\ldots,f_\ell$ is relatively prime to $x^m-1$ modulo $2$.
Theorem 3.2 weakens this assumption by requiring only that the
collection $f_1,\ldots,f_\ell$ be jointly relatively prime to it over
$\mathbb{F}_2$. Neither condition is necessary in general.

For an odd positive integer $m$, a more precise criterion can be
obtained from the factorization of $x^m-1$ over $\mathbb{F}_2$. Since
$m$ is odd, $x^m-1$ is squarefree over $\mathbb{F}_2$. The corresponding
factorization over $\mathbb{Z}_4$ yields a Chinese remainder
decomposition of
\[
R=\mathbb{Z}_4[x]/\langle x^m-1\rangle
\]
into finite chain rings. We use this decomposition to examine the
construction componentwise. Theorem 4.1 gives a necessary and
sufficient condition for the QC code $C$ to have the same type as the
cyclic code $C_g$. The criterion depends only on those irreducible
factors for which the image of $g$ in the corresponding component is
nonzero.

The same decomposition also gives a criterion for freeness. In
particular, Corollary 4.2 characterizes when the resulting QC code is a
free $\mathbb{Z}_4$-module.

Examples for $m=3$ are given in Section 5 to illustrate the differences
between the sufficient conditions of Section 3 and the necessary and
sufficient condition of Section 4. In Section 6, we present a number of QC codes obtained by a computational search guided by Theorem 4.1. Their
parameters and minimum Lee distances are compared with currently best known codes listed in the database (see \url{http://quantumcodes.info/Z4}). We have found 36 codes with better minimum Lee distances than the comparable best known codes. Finally, Section 7 discusses some potential
directions for further work.

\section{Preliminaries}
\label{sec:prelim}

Throughout this paper, $m$ is an odd positive integer and
$C_g = \langle g(x)\rangle$ is a cyclic code of length $m$ over
$\mathbb{Z}_4$, that is, an ideal of
\[
R := \mathbb{Z}_4[x]/\langle x^m-1\rangle
\]
generated by a single polynomial $g(x)$. As a finitely generated
$\mathbb{Z}_4$-module, $C_g$ decomposes as
\[
C_g \cong \mathbb{Z}_4^{k_1} \oplus \mathbb{Z}_2^{k_2}
\]
for unique $k_1,k_2\geq0$, so it has a well-defined type
$4^{k_1}2^{k_2}$ and cardinality
\[
|C_g|=4^{k_1}2^{k_2}.
\]

\noindent Our main objects of study in this work are $1$-generator QC codes of index $\ell$ of the form
\[
C =
\big\langle
\big(f_1(x)g(x),\dots,f_\ell(x)g(x)\big)
\big\rangle
\]
where $f_i(x)\in\mathbb{Z}_4[x]$. We write $\bar h(x)$, or simply
$h(x)$ when the context is clear, for the reduction of
$h(x)\in\mathbb{Z}_4[x]$ modulo $2$, regarded as an element of
$\mathbb{F}_2[x]$. Similarly, we write $\bar f(x)$ and $\bar f$ interchangeably for a polynomial when the variable is clear from context.

\begin{fact}[Squarefreeness]
\label{fact2.1}
For an odd positive integer $m$, $x^m-1$ is squarefree in
$\mathbb{F}_2[x]$.
\end{fact}

\begin{proof}
The formal derivative of $x^m-1$ is $mx^{m-1}$. Since $m$ is odd,
$m\equiv1\pmod2$, so in $\mathbb{F}_2[x]$ this derivative equals
$x^{m-1}$. As $x^m-1$ has nonzero constant term
$(-1=1\ne0)$, we have
\[
\gcd(x^m-1,x^{m-1})=1
\]
in $\mathbb{F}_2[x]$. A polynomial over a field is squarefree if and
only if it is coprime to its formal derivative, so $x^m-1$ is
squarefree over $\mathbb{F}_2$.
\end{proof}

\noindent This elementary fact of squarefreeness is a critical part of the argument below. By
Fact 2.1, we may factor
\[
x^m-1=\bar p_1(x)\cdots\bar p_r(x)
\quad\text{in }\mathbb{F}_2[x],
\]
where $\bar p_i$'s are distinct and irreducible. By Hensel's lemma,
since $x^m-1$ is separable over $\mathbb{F}_2$, each $\bar p_i$ lifts
uniquely to a monic factor
$p_i(x)\in\mathbb{Z}_4[x]$ of $x^m-1$ reducing to $\bar p_i$ modulo
$2$, and
\[
x^m-1=p_1(x)\cdots p_r(x)
\quad\text{in }\mathbb{Z}_4[x],
\]
with the $p_i$'s pairwise coprime in $\mathbb{Z}_4[x]$. By the Chinese
Remainder Theorem, this gives a ring isomorphism
\[
R:=\mathbb{Z}_4[x]/\langle x^m-1\rangle
\cong
\prod_{i=1}^r R_i,
\qquad
R_i:=\mathbb{Z}_4[x]/\langle p_i(x)\rangle.
\]

\begin{lemma}
\label{lem:chain}
With the notation above, each $R_i$ is a finite chain ring (see~\cite{NortonSalagean2000}), with the unique maximal ideal
$\mathfrak m_i=2R_i$ and the ideal lattice
\[
0\subsetneq2R_i\subsetneq R_i.
\]
Consequently every $u\in R_i$ can be written as
\[
u=2^v w
\]
with $w\in R_i^\times$ and $v\in\{0,1,2\}$, where $v=2$ corresponds
to $u=0$.
\end{lemma}

\begin{proof}
Write $d_i=\deg p_i$. Reducing
\[
R_i=\mathbb{Z}_4[x]/\langle p_i(x)\rangle
\]
modulo $2$ gives
\[
R_i/2R_i
\cong
\mathbb{F}_2[x]/\langle \bar p_i(x)\rangle
=:\mathbb{F}_{2^{d_i}},
\]
which is a field since $\bar p_i$ is irreducible. Thus, $2R_i$ is a maximal ideal.

Let $u\in R_i\setminus2R_i$. Then $u$ is nonzero in the field
$R_i/2R_i$, so there exists $v\in R_i$ with
\[
uv=1+2t
\]
for some $t\in R_i$. Since $R_i$ is a $\mathbb{Z}_4$-algebra,
$4=0$ in $R_i$, so
\[
(2t)^2=4t^2=0.
\]
Hence $1+2t$ is a unit with inverse $1-2t$, and therefore $u$ is a
unit. Consequently,
\[
R_i\setminus2R_i=R_i^\times.
\]

\noindent Now let $I\subseteq R_i$ be a nonzero ideal with $I\neq R_i$. Since
$I\neq R_i$, it contains no unit, so
\[
I\subseteq2R_i.
\]
Since $I\neq0$, choose $0\neq2t\in I$. Then $2t\neq0$ forces
$t\notin2R_i$; hence $t$ is a unit. Therefore
\[
2=(2t)t^{-1}\in I,
\]
so $2R_i\subseteq I$. Thus $I=2R_i$, proving the chain property. The valuation statement is immediate: $v=0$ if $u$ is a unit,
$v=1$ if $u\in2R_i\setminus\{0\}$, and $v=2$ if $u=0$.
\end{proof}

\begin{lemma}
\label{lem:ideals}
Let $R_i$ be as in Lemma 2.2, and let
$b_1,\dots,b_\ell\in R_i$. Then
\[
\langle b_1,\dots,b_\ell\rangle=
\begin{cases}
R_i,&\text{if some }b_j\text{ is a unit},\\
2R_i,&\text{if no }b_j\text{ is a unit but some }b_j\neq0,\\
0,&\text{if }b_1=\cdots=b_\ell=0.
\end{cases}
\]
\end{lemma}

\begin{proof}
By Lemma 2.2, the ideals of $R_i$ form the chain
\[
0\subsetneq2R_i\subsetneq R_i.
\]
Hence the sum
\[
\langle b_1\rangle+\cdots+\langle b_\ell\rangle
\]
is simply the largest of the ideals $\langle b_j\rangle$ appearing
in this chain. Each $\langle b_j\rangle$ equals $R_i$, $2R_i$, or $0$
depending on whether $b_j$ is a unit, a nonzero non-unit, or zero.
\end{proof}

\begin{remark}[Ideals and their cardinalities]
\label{rem:ideals}
In our setting, the ideals of the finite chain ring $R_i$ are
\[
0\subsetneq2R_i\subsetneq R_i.
\]
Writing $d_i=\deg p_i$, we have
\[
|R_i|=4^{d_i},
\qquad
|2R_i|=2^{d_i},
\qquad
|0|=1.
\]
Indeed, $R_i$ is a free $\mathbb{Z}_4$-module of rank $d_i$, while
$2R_i$ is a vector space over $\mathbb{F}_2$ of dimension $d_i$.
Thus the three ideals have distinct cardinalities.

Since an $R_i$-module isomorphism is, in particular, a bijection,
two ideals of $R_i$ can be isomorphic only if they have the same
cardinality. Consequently, among the ideals of $R_i$, two ideals are
isomorphic as $R_i$-modules if and only if they are equal as sets.

This observation is useful in the proof of Theorem 4.1.
There we have
\[
C_i\subseteq\langle g_i\rangle
\]
for every $i$. Hence, if
\[
|C_i|=|\langle g_i\rangle|,
\]
then
\[
C_i=\langle g_i\rangle.
\]
Thus equality of the global cardinalities can be used to force equality
componentwise.
\end{remark}

\begin{remark}
We use both $f(x)$ and $f$ to denote a polynomial.
\end{remark}

\section{Sufficient Conditions for Type Preservation}
\label{sec:sufficient}

\noindent
Our first result shows that it suffices for a single $f_j$ to be
coprime to $x^m-1$ over $\mathbb{F}_2$ for type preservation.

\begin{theorem}
\label{thm:single}
If there exists
$f_{j_0}(x)\in\{f_1(x),\dots,f_\ell(x)\}$ such that
\[
\gcd\big(\bar f_{j_0}(x),x^m-1\big)=1
\]
in $\mathbb{F}_2[x]$, then $C$ has the same type
$4^{k_1}2^{k_2}$ as $C_g$.
\end{theorem}

\begin{proof}
Define
\[
\varphi:C_g\longrightarrow C,
\qquad
\varphi(a(x)g(x))=
(a(x)f_1(x)g(x),\dots,a(x)f_\ell(x) g(x)),
\quad a(x)\in R.
\]
The map $\varphi$ is well defined: if $a(x)g(x)=b(x)g(x)$, then
$(a(x)-b(x))g(x)=0$, and hence
\[
(a(x)-b(x))f_i(x) g=f_i(a-b)g=0
\]
for every $i$. It is clearly an $R$-module homomorphism and is
surjective.

It remains to show that $\varphi$ is injective. Suppose
$\varphi(a(x)g(x))=0$. Then
\[
a(x)f_i(x) g(x)=0
\]
for every $i$, and in particular
\[
a(x)f_{j_0}(x)g(x)=0.
\]
Since
\[
\gcd(\bar f_{j_0}(x),x^m-1)=1 \text{ in }\mathbb{F}_2[x],
\]
$f_{j_0}(x)$ is a unit in $R$. Therefore
\[
a(x)g(x)=0.
\]
Thus $\ker\varphi=0$, and $\varphi$ is an isomorphism. Hence
$C\cong C_g$ and $C$ has the same type as $C_g$.
\end{proof}

\noindent Now we give another sufficient condition which is a generalization of the previous theorem.

\begin{theorem}
\label{thm:joint}
Let $C_g=\langle g(x)\rangle$ be a cyclic code of odd length $m$ over
$\mathbb{Z}_4$ of type $4^{k_1}2^{k_2}$, and let
\[
C=\langle(f_1(x)g(x),\dots,f_\ell(x)g(x))\rangle.
\]
If
\[
\gcd\big(
\bar f_1(x),\dots,\bar f_\ell(x),x^m-1
\big)=1
\]
in $\mathbb{F}_2[x]$, then $C$ has the same type
$4^{k_1}2^{k_2}$ as $C_g$.
\end{theorem}

\begin{proof}
As before let
\[
R=\mathbb{Z}_4[x]/\langle x^m-1\rangle.
\]
By hypothesis and the B\'ezout identity over $\mathbb{F}_2$, there exist
$b_1(x),\dots,b_\ell(x)\in\mathbb{F}_2[x]$ such that
\[
b_1(x)\bar f_1(x)+\cdots+b_\ell(x)\bar f_\ell(x)
\equiv1\pmod{x^m-1}.
\]

\noindent Lift each $b_i(x)$ to $B_i(x)\in\mathbb{Z}_4[x]$ and set
\[
E(x)=1-\sum_{i=1}^{\ell}B_i(x)f_i(x)\in R.
\]
Then
\[
E(x)=2h(x)
\]
for some $h(x)\in R$, so
\[
\sum_{i=1}^{\ell}B_i(x)f_i(x)=1-2h(x).
\]
Put
\[
a_i(x)=(1+2h(x))B_i(x).
\]
Then
\[
\sum_{i=1}^{\ell}a_i(x)f_i(x)
=
(1+2h(x))(1-2h(x))
=
1-4h(x)^2
=
1.
\]

\noindent Thus
\[
a_1(x)f_1(x)+\cdots+a_\ell(x)f_\ell(x)=1 \text{ in }R.
\]

\noindent With this identity, define
\[
\varphi:C_g\to C,
\qquad
\varphi(ag)=(af_1g,\dots,af_\ell g).
\]
If $\varphi(ag)=0$, then
\[
0
=
\sum_{i=1}^{\ell}a_i(x)(af_i g)
=
a\left(\sum_{i=1}^{\ell}a_i(x)f_i(x)\right)g
=
ag.
\]
Hence $\ker\varphi=0$, so
\[
C\cong C_g.
\]
Therefore $C$ has the same type as $C_g$.
\end{proof}

\section{A Necessary and Sufficient Condition}
\label{sec:necsuff}

Both prior results share a common shortcoming: they require $f_1,\dots,f_\ell$
to be jointly coprime to the entire polynomial $x^m-1$, when only the
irreducible factors $p_i(x)$ for which the image of $g(x)$ in the
corresponding component is nonzero can affect the type of $C$.
Localizing at each factor $p_i(x)$ makes this precise.

Recall the isomorphism
\[
\Phi:R\longrightarrow\prod_{i=1}^rR_i,
\qquad
\Phi(h)=(h_1,h_2,\dots,h_r).
\]
Write $g_i$ for the image of $g(x)$ in $R_i$, and set
\[
S=\{i:g_i\neq0\}.
\]

By Lemma 2.2, every nonzero $g_i$ is either a unit of
$R_i$ or a nonzero element of $2R_i$. Splitting $S$ accordingly, we define:
\[
S_1=
\{i\in S:g_i\text{ is a unit of }R_i\},
\]
\[
S_2=
\{i\in S:g_i\in2R_i\setminus\{0\}\}.
\]
Thus
\[
S=S_1\sqcup S_2.
\]

\noindent Applying $\Phi$ to $C_g=\langle g\rangle$ and using
Lemma 2.3, we obtain
\[
\Phi(C_g)=\prod_i\langle g_i\rangle.
\]
Counting elements gives
\[
k_1=\sum_{i\in S_1}d_i,
\qquad
k_2=\sum_{i\in S_2}d_i,
\qquad
d_i=\deg p_i(x).
\]
In particular, $C_g$ is free over $\mathbb{Z}_4$ exactly when
$S_2=\varnothing$. Now, we state the main theorem of this paper.

\begin{theorem}
\label{thm:main}
With notation as above,
\[
\gcd\left(
\bar f_1(x),\dots,\bar f_\ell(x),
\prod_{i\in S}\bar p_i(x)
\right)=1
\]
in $\mathbb{F}_2[x]$ if and only if $C$ has the same type
$4^{k_1}2^{k_2}$ as $C_g$.
\end{theorem}

\begin{proof}
Applying $\Phi$ to the generating tuple of $C$ gives
\[
\Phi(C)=\prod_i C_i,
\]
where
\[
C_i=
\langle f_{1,i}g_i,\dots,f_{\ell,i}g_i\rangle
\subseteq R_i.
\]
For $i\notin S$, we have
\[
C_i=0=\langle g_i\rangle.
\]

Fix $i\in S$. If $i\in S_1$, multiplication by the unit $g_i$ is a
bijection of $R_i$. Hence
\[
C_i
=
g_i\langle f_{1,i},\dots,f_{\ell,i}\rangle.
\]
By Lemma 2.3,
\[
C_i=\langle g_i\rangle
\]
precisely when some $f_{j,i}$ is a unit.

\noindent If $i\in S_2$, then $g_i$ generates $2R_i$. If some $f_{j,i}$ is a
unit, then $f_{j,i}g_i$ generates $2R_i$. If every $f_{j,i}$ lies in $2R_i$, then
\[
f_{j,i}g_i\in4R_i=0
\]
for every $j$, so $C_i=0$. Thus, in both cases,
\[
C_i=\langle g_i\rangle
\iff
\text{some }f_{j,i}\text{ is a unit of }R_i.
\]
By the unit criterion,
\[
\text{some }f_{j,i}\text{ is a unit}
\iff
\bar p_i(x)\nmid\bar f_j(x)
\quad\text{for some }j.
\]

\noindent Suppose now that $C$ has the same type as $C_g$. Then
\[
|C|=|C_g|.
\]
If the above condition failed for some $i_0\in S$, then
\[
|C_{i_0}|
\le
\frac{|\langle g_{i_0}\rangle|}{2^{d_{i_0}}},
\]
while
\[
C_i\subseteq\langle g_i\rangle
\]
for every other $i$. Therefore
\[
|C|
\le
\frac{|C_g|}{2^{d_{i_0}}}
<
|C_g|,
\]
a contradiction. Hence
\[
C_i=\langle g_i\rangle
\]
for every $i\in S$.

Conversely, if the condition holds for every $i\in S$, then
\[
C_i=\langle g_i\rangle
\]
for every $i$, and therefore
\[
\Phi(C)=\Phi(C_g).
\]
Since $\Phi$ is an isomorphism,
\[
C=C_g.
\]
In particular, $C$ and $C_g$ have the same type.

Finally, since the $\bar p_i$'s are distinct irreducible factors, the condition that for all $i \in S$ there exists $j$ such that $\bar p_i \nmid \bar f_j$ is equivalent to
\[
\gcd\left(\bar f_1, \dots, \bar f_\ell, \prod_{i \in S} \bar p_i\right) = 1.
\]
This proves the theorem.
\end{proof}

This settles the open problem posed in
\cite{AydinLuOnta2023}: we obtain a necessary and sufficient condition
for the $1$-generator quasi-cyclic code $C$ to have the same type as
the underlying cyclic code $C_g$.

\begin{corollary}[Freeness criterion]
\label{cor:free}
Suppose $f_1(x),\dots,f_\ell(x)$ satisfy the condition of
Theorem 4.1, so that $C\cong C_g$. Then $C$ is a free
$\mathbb{Z}_4$-module if and only if
\[
\gcd\left(
\bar g(x),
\prod_{i\in S}\bar p_i(x)
\right)=1
\]
in $\mathbb{F}_2[x]$, equivalently if and only if
$S_2=\varnothing$.
\end{corollary}

\begin{proof}
Since $C\cong C_g$, the code $C$ is free exactly when $C_g$ is free,
that is, when
\[
k_2=\sum_{i\in S_2}d_i=0.
\]
Hence
\[
S_2=\varnothing.
\]
For $i\in S$, the condition $i\in S_2$ is equivalent to
\[
g_i\in2R_i\setminus\{0\},
\]
which is in turn is equivalent to
\[
\bar p_i(x)\mid\bar g(x).
\]
The desired gcd criterion follows.
\end{proof}

\section{Examples and Non-Examples}
\label{sec:examples}

The three theorems above form a chain of increasingly permissive
conditions. The smallest odd length for which the distinctions between these conditions arise is $m=3$.

Take $m=3$, so
\[
x^3-1=(x-1)(x^2+x+1) \text{ over } \mathbb{Z}_4.
\]
Let $p_1(x)=x-1,   p_2(x)=x^2+x+1$. The B\'ezout identity
\[
(x+2)(x-1)+3(x^2+x+1)=1
\]
holds in $\mathbb{Z}_4[x]$.

We work with two seed codes:
\[
C_g=\langle g(x)\rangle,
\qquad
g(x)=x^2+x+3,
\]
and
\[
C_{g'}=\langle g'(x)\rangle,
\qquad
g'(x)=2x^2+2x.
\]

A direct check shows that $C_g$ has type $4^1 2^2$; both
$p_1$- and $p_2$-factors are \emph{live}: reducing $g$ modulo
$p_1$ and $p_2$ gives $g_1=1$, a unit of $R_1$, and
$g_2=2$, a nonzero non-unit of $R_2$, so both components are
nonzero. Meanwhile, $C_{g'}$ has type $4^0 2^2$ with only the
$p_2$-factor live, since $g'_1=0$ but $g'_2=2\neq0$ (so
$S=\{2\}$).

\begin{example}[Theorem 3.1, positive example]

Take $f_1(x)=x$. Since $x^3\equiv1$ in $R$, $x$ is a unit, and the
hypothesis of Theorem~3.1 holds. Then
\[
C=\langle xg(x)\rangle
\]
still has type $4^1 2^2$.
\end{example}

\begin{example}[Theorem 3.1 fails, conclusion still holds]
Take
\[
f_1(x)=x-1,
\qquad
f_2(x)=x^2+x+1.
\]
Neither $f_1$ nor $f_2$ is coprime to $x^3-1$, so the hypothesis of
Theorem 3.1 fails. Yet the type is preserved:
\[
C=\langle(x-1)g,(x^2+x+1)g\rangle
\]
still has type $4^1 2^2$.
\end{example}

\begin{example}[Theorem 3.2, positive example]
Continuing with the previous example,
\[
\gcd(\bar f_1,\bar f_2,x^3-1)=1
\]
in $\mathbb{F}_2[x]$. Thus, Theorem 3.2 applies and correctly predicts
the type $4^1 2^2$.
\end{example}

\begin{example}[Theorem 3.2 fails, the conclusion still holds]
Take $\ell = 1$ and $f_1(x)=x-1$.
Apply it to $C_{g'}$. Theorem~3.2 fails because
\[
\gcd(\bar f_1,x^3-1)=x+1\neq1.
\]
Yet
\[
C=\langle(x-1)g'\rangle
\]
still has type $4^0 2^2$, the same as $C_{g'}$.
\end{example}

\begin{example}[Theorem 4.1, condition fails]
Take
\[
f_1(x)=x^2+x+1=p_2(x)
\]
and apply it to $C_g$. Then the condition of Theorem 4.1 fails.
Indeed,
\[
C=\langle p_2(x)g(x)\rangle
\]
has type $4^1 2^0$, whereas $C_g$ has type $4^1 2^2$.
\end{example}

\begin{example}[Theorem 4.1, second failure]
Take again
\[
f_1(x)=x^2+x+1=p_2(x),
\]
now applied to $C_{g'}$. Since $S=\{2\}$, the condition of
Theorem 4.1 fails. Indeed,
\[
f_1(x)g'(x)\equiv0\pmod{x^3-1},
\]
so
\[
C=\{0\},
\]
of type $4^0 2^0$, whereas $C_{g'}$ has type $4^0 2^2$.
\end{example}

\section{Record Breaking Codes}
\label{sec:computational}

The structural results developed above, in particular Theorem~4.1,
can be used to guide a systematic computer search for quasi-cyclic
codes over
$\mathbb{Z}_4$.
Applying this criterion to suitable families of defining
polynomials, we found 36 quaternary linear codes with new parameters that, to the best of our knowledge, do not appear in the existing literature. We present these codes and their parameters in this section. Every code below was found by a targeted search over odd values of $m$, using QC codes of the form given in Theorem~4.1. For each cyclic seed $C_g$ we computed the live set $S=\{i:g_i\neq0\}$ from the CRT factorization of $x^m-1$, then searched over $f_1(x),\dots,f_\ell(x)$ satisfying
\[
\gcd\left(\bar f_1,\dots,\bar f_\ell,\prod_{i\in S}\bar p_i\right)=1 \text{ in } \mathbb{F}_2[x],
\]
so that Theorem~4.1 guarantees that $C\cong C_g$ and hence that $C$ has the same type as the seed.

The parameters and the minimum Lee distances $d_L$ of the codes below have been computed using Magma software~\cite{Magma1997}. In Table~\ref{tab:section6-params} below, we present a subset  of the 36 new codes we have found. The table gives, for each code, the parameters $[n,k_1,k_2,d_L]$, the index $\ell$, whether the binary Gray-map image is linear or nonlinear, the best known minimum distance for the corresponding length and type, together with the improvement our code achieves over it (i.e., $d_L^{\text{best known}}(+\Delta)$ where $\Delta = d_L - d_L^{\text{best known}}$). For comparison, we used the parameters of the best known codes in $\mathbb{Z}_4$ given in the database (see \url{http://quantumcodes.info/Z4}).

\renewcommand{\arraystretch}{1.25}
\footnotesize
\setlength{\tabcolsep}{6pt}

\begin{longtable}{@{}c c c c c l@{}}
\caption{Parameters of the new Quaternary codes found via Theorem 4.1, all at odd $m$.}
\label{tab:section6-params}\\
\toprule
\textbf{ID} & \textbf{$[n,k_1,k_2,d_L]$} & \textbf{$\ell$} & \textbf{Gray image} & \textbf{Best known $d_L$} & \textbf{Improvement} \\
\midrule
\endfirsthead

\multicolumn{6}{c}{\tablename\ \thetable\ -- continued from previous page}\\
\toprule
\textbf{ID} & \textbf{$[n,k_1,k_2,d_L]$} & \textbf{$\ell$} & \textbf{Gray image} & \textbf{Best known $d_L$} & \textbf{Improvement} \\
\midrule
\endhead

\midrule
\multicolumn{6}{r}{continued on next page}\\
\endfoot

\bottomrule
\endlastfoot

C1 & $[18,1,8,12]$   & $2$ & linear    & $8$  & $+4$  \\
C2 & $[21,6,1,14]$   & $3$ & nonlinear & $12$ & $+2$  \\
C3 & $[18,8,1,6]$    & $2$ & nonlinear & $4$  & $+2$  \\
C4 & $[44,1,10,30]$  & $4$ & linear    & $8$  & $+22$ \\
C5 & $[77,1,10,60]$  & $7$ & linear    & $14$ & $+46$ \\
C6 & $[26,1,12,12]$  & $2$ & linear    & $8$  & $+4$  \\
C7 & $[72,2,6,60]$   & $8$ & nonlinear & $32$ & $+28$ \\
C8 & $[66,1,10,48]$  & $6$ & linear    & $44$ & $+4$  \\
C9 & $[88,1,10,70]$  & $8$ & linear    & $16$ & $+54$ \\

\end{longtable}
\normalsize
For the seed polynomial $g(x)$ and the polynomials $f_1(x),\ldots,f_\ell(x)$ used to construct all 36 codes, see this \href{https://github.com/glitch1729/Better-Quaternary-Codes/raw/main/Quaternary_codes.pdf}{supplementary file}. These codes have
also been submitted to the database of $\mathbb{Z}_4$ codes
at \url{http://quantumcodes.info/Z4}

\section{Open Directions}
\label{sec:open}
This paper settles the open problem stated in
\cite{AydinLuOnta2023} concerning conditions under which a
$1$-generator quasi-cyclic code over $\mathbb{Z}_4$ has the same type
as its underlying cyclic code. The methods in this paper rely on \(m\) being odd, which ensures that \(x^m-1\) is squarefree over \(\mathbb F_2\) (Fact 2.1) and hence allows \(R\) to decompose via the CRT into a product of finite chain rings. For even $m$,
$x^m-1$ is not squarefree over $\mathbb{F}_2$, so the local components
of $R$ are no longer necessarily finite chain rings, and the argument
behind Theorem 4.1 does not directly apply. Extending the criterion to
even length is a natural next step.

\end{document}